\documentclass[12pt,a4paper]{article}
\usepackage{amscd}
\usepackage{amssymb}
\usepackage{amsmath}
\usepackage{amsthm}
\usepackage{latexsym}
\usepackage{upref}
\usepackage{graphicx}
\usepackage{epsfig}
\usepackage{graphics}
\usepackage{mathrsfs}
\usepackage{makeidx}
\usepackage{float}
\usepackage{leqno}
\usepackage{url}
\usepackage{soul}
\usepackage{caption}
\usepackage[usenames,dvipsnames]{color}
\renewcommand{\title}[1]{\begin{center}{\large{#1}}\end{center}}
\renewcommand{\author}[1]{{\bf{#1}}}
\newcommand{\affiliation}[1]{{\textsc{{#1}}}}
\newcommand{\email}[1]{\texttt{#1}}

\begin{document}
\title{Two-Level Fingerprinting Codes: Non-Trivial Constructions}

\begin{center}
\author{Penying Rochanakul}\\
\affiliation{Royal Holloway, University of London}\\
\email{P.Rochanakul@rhul.ac.uk}
\end{center}
%
%
%
\newtheorem{theorem}{Theorem}
\newtheorem{lemma}[theorem]{Lemma}
\newtheorem{proposition}[theorem]{Proposition}
\newtheorem{corollary}[theorem]{Corollary}
\theoremstyle{definition}
\newtheorem{definition}{Definition}
\newtheorem{example}{Example}
\theoremstyle{remark}
\newtheorem{remark}[theorem]{Remark}
\newtheorem{mycon}{Construction}
%
\newenvironment{pf}[1][Proof]{\begin{trivlist}
\item[\hskip \labelsep {\bfseries #1}]}{\end{trivlist}}
\newenvironment{dpf}[1][Direct proof of Corollary \ref{col} in the case $T = 2$.]{\begin{trivlist}
\item[\hskip \labelsep {\bfseries #1}]}{\end{trivlist}}
\newenvironment{pfm}[1][Proof of Theorem \ref{mainw}.]{\begin{trivlist}
\item[\hskip \labelsep {\bfseries #1}]}{\end{trivlist}}
\newcommand{\desc}{\mathrm{desc}}



\begin{abstract} 
We extend the concept of two-level fingerprinting codes, introduced by Anthapadmanabhan and Barg (2009) in context of traceability (TA) codes  \cite{Antha09}, to other types of fingerprinting codes, namely identifiable parent property (IPP) codes, secure-frameproof (SFP) codes, and frameproof (FP) codes. We define and propose the first explicit non-trivial construction for two-level IPP, SFP and FP codes.
\end{abstract}


%
%
\section{Introduction}


Fingerprinting codes have been studied for more than 15 years due to their applications in digital data copyright protection and their combinatorial interest. In order to trace back to the source of pirate piece of digital data, the data distributor hides a unique mark called a \textit{fingerprint} in each legal copy. 
The fingerprint can be viewed as a codeword length $\ell$, which each component is picked from alphabet $Q$ of cardinality $q$. Since all fingerprints are unique, if a naive user tries to distribute his or her copy illegally, the fingerprint will lead back to its owner right away. However, if a group of users form a \textit{coalition} and pool their copies together, 
it becomes a lot more complicated to trace back to the coalition members. It is assumed that the coalition can only produce data fingerprinted by words such that each component agrees with at least one of the codewords in the coalition: an element of the \textit{descendant} set of the codewords.
We define descendant again clearly in the next section. Once such a descendant is given, TA codes are designed to trace back at least one member of the coalition, provided that the size of the coalition does not exceed a certain parameter $t$. The user with the codeword that is most similar to the given descendant, is guaranteed to be a member of the coalition. However, there is no guarantee the tracing result is correct if the size of the coalition is greater than $t$. The first attempt to gain more information about the coalition when its size is bigger than $t$ was by Anthapadmanabhan and Barg (2009) in \cite{Antha09} who defined the concept of two-level fingerprinting codes; from now on we refer to their codes as two-level traceability codes. In Anthapadmanabhan and Barg's setting, the users are divided into various groups; for instance, 
by dividing the distribution area into several geographic regions, and collecting users from the same region into one group. As in traditional (one-level) 
codes, when given an illegal copy produced by a coalition of users, the codes can be used to identify one of the guilty users if the coalition size is less than or equal to a 
certain threshold t. However, even when the coalition is of a larger size $T$, where $T>t$, the decoder still provides partial information by tracing 
one of the groups containing a guilty user. The group that contains the codeword that is closest to the illegal copy must contain at least one member of the coalition. However, they provide no explicit construction of two-level traceability codes apart from using random codes when $t=1$ or $t=2$. 
IPP, SFP and FP codes have been widely studied as replacements for, or weakenings of, TA codes \cite{Boneh98, Stinson97}.
In this paper, we aim to explore TA, IPP, SFP and FP codes in the two-level context. Although TA codes have a very efficient algorithm for traitor tracing and very interesting in terms of applications, there is no explicit construction for two-level traceability codes available, except in the trivial case.  
IPP codes can replace TA codes, since they also have tracing ability when a coalition under a threshold $t$. Moreover, under the same parameters, such as length and alphabet size, IPP codes are much bigger. Therefore, it is important to study two-level IPP codes. We also study SFP and FP codes, which focus on preventing innocent users from being framed, and analyze their interesting combinatorial properties. We define each type of codes properly in the next section. Our main result (Theorems \ref{2lfp}, \ref{2lsfp}, and \ref{2lipp}) is a construction for two-level IPP, SFP and FP codes when the number of the groups is small, i.e. less than or equal to the alphabet size. This is actually the first explicit non-trivial constructions for two-level fingerprinting codes. 

More details on fingerprinting codes' applications and combinatorial properties can be found in a survey by Blackburn \cite{Simon03}. 
We also suggest reading these prior works for more in-depth details: see \cite{Naor94} and  \cite{Chor94} for the concept of traitor tracing along with traceability;  see  \cite{Simon09},  \cite{Staddon00}, and  \cite{Stinson98} for further combinatorial properties of TA codes; more details about FP and SFP codes are available in  \cite{Boneh98} and  \cite{Stinson97}, respectively.

The different types of codes are defined in Section 2 and the relationships between the definitions are explored. In Section 3, we provide our proposed construction for two-level codes. 
Section 4 is a short conclusion containing open problems.
%
%
\section{Definitions and Relationships}
In this section, we restate the definitions of the classical one-level codes we are interested in, then define the corresponding two-level codes.

\subsection{One-Level Codes}
Before giving the codes' definitions, it is necessary to introduce some notation. Let $C$ be a code of length $\ell$ on an alphabet $Q$ of (finite) size $q$. Then, we issue data with a hidden different fingerprint from $C$ to each user. The elements of $C$ are called \textit{codewords}. The \textit{hamming distance} between codewords $x, y$ will be written as $d_H(x,y)$. For any $X \subseteq Q^{\ell}$ and $y \in Q^{\ell}$, let $d(X)$ denote the \textit{minimum distance} of $X$ and let $d_H(X, y) = \displaystyle\min_{x\in X}  d_H(x, y)$. 

We call any subset $C \subseteq Q^{\ell}$, a $q$-ary length $\ell$  
code.
For each $x \in Q^{\ell}$, we write $x_i$ for the $i$th component of $x$. For $X \subseteq Q^{\ell}$, we define \textit{the set of descendants} of $X$ to be the subset $\desc(X) \subseteq  Q^{\ell}$ given by
\begin{align*}
\desc(X) = \{d \in Q^{\ell} : \forall i \in \{1, 2, ..., \ell\},  \exists x \in X \text{ such that } d_i = x_i\}.
\end{align*}
For example, let $P = \{1100, 2102, 1122\}$, then
\begin{align*}
\desc(P) = \{1100, 1102, 1120, 1122, 2100, 2102, 2120, 2122\}.
\end{align*}
Let $t$ be a positive integer. For a code $C$ define the $t$-\textit{descendant code} of $C$, denoted $\desc_t(C)$, by
\begin{align*}
\desc_t(C) = \displaystyle\bigcup_{\substack{X \subseteq C \\ |X| \leq t}} \desc(X).
\end{align*}

\begin{definition} \label{1ld}Let $C$ be an $q$-ary length $\ell$ code and let $t$ be a positive integer.
\begin{enumerate}
  \item[(i)] {$C$ has the $t$-\textit{frameproof} property (or is $t$-FP) if for all $X \subseteq C$ such that $|X| \leq t$, it holds that}
      \begin{align*}
            \desc(X) \cap C \subseteq X
      \end{align*}
  \item[(ii)] {$C$ has the $t$-\textit{secure-frameproof} property (or is $t$-SFP) if for all $X_0, X_1 \subseteq C$ of size at most $t$, if $\desc(X_0) \cap \desc(X_1) \not = \emptyset$, then $X_0 \cap X_1 \not = \emptyset$.}
  \item[(iii)] {$C$ has the $t$-\textit{identifiable parent property} (or is $t$-IPP) if for all $x \in \desc_t(C)$, it holds that}
      \begin{align*}
            \displaystyle\bigcap_{\substack{X \subseteq C : |X| \leq t \\ x \in \desc(X)}} X \not = \emptyset.
      \end{align*}
  \item[(iv)] {$C$ has the $t$-\textit{traceability} property (or is $t$-TA) if for all $X \subseteq C$ that $|X| \leq t$ and for all $x \in \desc(X)$, for all $z \in C$ with $d_H(x,z)$ minimal
, $z \in X$.}

\end{enumerate}
\end{definition}

Presuming that the coalition size is at most $t$. None of the coalition under FP codes can produce a codeword that is not belong to the coalition. In SFP codes, the descendant sets of codewords from two disjoint coalitions are always disjoints. All the coalition of IPP codes that can produce the same words, own at least one codeword in common. Lastly, in TA codes, given a descendant of any coalition under TA code, a codeword that is most similar to the given descendant, is always a member of the coalition.

It is not difficult to check from the definitions that the relationships among different types of codes are as follows; $t$-TA codes are $t$-IPP codes, $t$-IPP codes are $t$-SFP codes and $t$-SFP codes are $t$-FP codes. 


\subsection{Two-Level Codes}
In this section, we extend the concept of two-level fingerprinting codes, introduced by Anthapadmanabhan and Barg (2009) in context of traceability (TA) codes  \cite{Antha09}, to IPP, SFP and FP codes. Then, we present an overview of relationships between different types of fingerprinting codes. 

The motivation behind the definitions of two-level fingerprinting codes is to gain more information from a word produced by a coalition that is bigger than the first threshold $t$. Consider the following scenario. The content distributor distributes the same number of distinct digital copies to each company. Then those companies assign each copy in their hands to an individual employer. Here each company is acting as a group in our two-level model. Once a piracy has occurred, apart from tracing back to an individual traitor or protecting an innocent user from being framed as can be achieved by one-level codes, one might be interested in just tracing back to a company that employs one of the coalition members and sue the whole company (two-level TA and two-level IPP), or to make sure that none of those companies can cooperate and frame the other innocent companies without getting one of their employee involved (two-level SFP and two-level FP). 

In traditional one-level codes, we assign a different fingerprint from $C$ to each user. Here $C$ is a code of length $\ell$ over an alphabet $Q$ of (finite) size $q$. In two-level codes, we partition a one-level code $C$ into $g$ disjoint sets of $p$ elements each, denoted by $C_1, C_2,...C_{g}$. Then we call $C$ a $q$-ary length $\ell$ two-level code, containing $g$ groups of size $p$. Hence, $C=C_1 \cup C_2 \cup ... \cup C_{g}$, we have $|C_i|=p$ for all $i,j \in \{1, 2, ..., g\}$, and $C_i \cap C_j = \emptyset$ when $i \not = j$. For a codeword $c \in C_i$, let $\mathcal{G}(c)$ represent its group index $i$.

\begin{definition} Let $C=C_1 \cup C_2 \cup ... \cup C_{g}$ be a two-level $q$-ary length $\ell$ code and let $T, t$ be positive integers where $T \geq t$.
\begin{enumerate}
  \item[(i)] {$C$ has the $(T,t)$-\textit{frameproof} property (or is $(T,t)$-FP) if}
\begin{enumerate}
  \item {$C$ is $t$-FP when viewed as a $q$-ary length $\ell$  
  code,} and
  \item {for all $X \subset C$ such that $|X| \leq T$ and for all $x \in \desc(X) \cap C$, $\mathcal{G}(x) \in  \mathcal{G}(X)$.}
\end{enumerate}
  \item[(ii)] {$C$ has the $(T,t)$-\textit{secure-frameproof} property (or is $(T,t)$-SFP) if}
\begin{enumerate}
  \item {$C$ is $t$-SFP when viewed as a $q$-ary length $\ell$
  code,} and
  \item {for all $X_0, X_1 \subseteq C$ of size at most $T$, if $\desc(X_0) \cap \desc(X_1) \not = \emptyset$, then $ \mathcal{G}(X_0) \cap  \mathcal{G}(X_1) \not = \emptyset$.}
\end{enumerate}
  \item[(iii)] {$C$ has the $(T,t)$-\textit{identifiable parent property} (or is $(T,t)$-IPP) if}
\begin{enumerate}
  \item {$C$ is $t$-IPP when viewed as a $q$-ary length $\ell$ 
   code,} and
  \item {for all $x \in \desc_{T}(C)$,}
  \begin{align*}
 \displaystyle\bigcap_{\substack{X \subseteq C : |X| \leq T \\ x \in \desc(X)}} \mathcal{G}(X) \not = \emptyset.
  \end{align*}
\end{enumerate}
  \item[(iv)] {$C$ has the $(T,t)$-\textit{traceability} property (or is $(T,t)$-TA) if}
\begin{enumerate}
  \item {$C$ is $t$-TA when viewed as a $q$-ary length $\ell$ 
   code,} and
  \item {for all $X \subseteq C$ that $|X| \leq T$ and for all $x \in \desc(X)$,}\\ for all $z \in C$ with $d_H(x,z)$ minimal (i.e. $d_H(x,z) = d_H(C,x)$), $\mathcal{G}(z) \in \mathcal{G}(X)$.
\end{enumerate}
\end{enumerate}
We refer to $(T,t)$-FP codes, $(T,t)$-SFP codes, $(T,t)$-IPP codes and $(T,t)$-TA codes as two-level codes, and refer to all codes in Definition \ref{1ld} as one-level codes.
\end{definition}

The relationships among different types of two-level codes and one-level codes can be illustrated as in the following diagram. 
\begin{table}[h]
\begin{center}
\begin{tabular}{c c c c c}
$T$-TA & $\Longrightarrow$ & $(T,t)$-TA & $\Longrightarrow$ & $t$-TA\\
$\Downarrow$&&$\Downarrow$& & $\Downarrow$\\
$T$-IPP & $\Longrightarrow$ & $(T,t)$-IPP & $\Longrightarrow$ & $t$-IPP\\
$\Downarrow$&&$\Downarrow$& & $\Downarrow$\\
$T$-SFP & $\Longrightarrow$ & $(T,t)$-SFP & $\Longrightarrow$ & $t$-SFP\\
$\Downarrow$&&$\Downarrow$& & $\Downarrow$\\
$T$-FP & $\Longrightarrow$ & $(T,t)$-FP & $\Longrightarrow$ & $t$-FP.\\
\end{tabular}
\end{center}
\end{table}

The proofs of the implications are straightforward from the codes' definitions. We state and prove the next theorem as an example. 

\begin{theorem}\label{2taisipp} A $(T, t)$-TA code is a $(T, t)$-IPP code.
\end{theorem}
\begin{proof} Let $C$ be a $(T, t)$-TA code. Accordingly, $C$ is $t$-TA when viewed as a one-level code. That is, $C$ is a $t$-IPP code.

Let $x \in \desc_{T}(C)$ and let $D_x$ be a set of all elements $z$ in $C$ that give the minimum value of $d_H(z,x)$.
Let $U \subseteq C$ be any subset of size at most $T$ such that $x \in \desc(U)$.
By the definition of two-level traceability code, we know that  for all $z \in \mathcal{G}D_x$, $\mathcal{G}(z) \in \mathcal{G}(U)$. Hence $\mathcal{G}(D_x) \subseteq \mathcal{G}(U)$. Then, we have
\begin{align*}
\mathcal{G}(D_x) \subseteq \displaystyle\bigcap_{\substack{X \subseteq C : |X| \leq T \\ x \in \desc(X)}} \mathcal{G}(X).
\end{align*}
Consequently,
\begin{align*}
\displaystyle\bigcap_{\substack{X \subseteq C : |X| \leq T \\ x \in \desc(X)}} \mathcal{G}(X) \not = \emptyset.
\end{align*}
Therefore, $C$ is a $(T, t)$-IPP code.
\end{proof}
The proof of the rest of the relationships between the different type of codes can be found in \cite{thesis}. 

As a consequence of the definitions of codes, all two-level fingerprinting codes possess the properties of one-level codes. So, it is natural to construct two-level fingerprinting codes from existing one-level fingerprinting codes. We propose our explicit construction in the next section.

%
%
\section{Code Constructions}
In this section, we aim to construct two-level fingerprinting codes from existing one-level fingerprinting codes. Our construction works provided the number of groups is at most the alphabet size. Our construction begins with a one-level code, removes some codewords, groups and modifies the remaining codewords. The results are two-level codes with are guaranteed to be at least half the size of their original codes. 

{Straight from the definitions, it is easy to see that the upper bounds on the size of one-level codes are also relevant to two-level codes. In fact, rather than using our construction, the following example is better for some special cases. The construction here is simple and provides two-level codes that meet the lower bounds of one-level codes.} 

\begin{example}\label{sepa}
Let $\ell$ be a positive  integer greater than 2, and let $Q$ a finite field of cardinality at least $\ell$. Let $\alpha_1, \alpha_2, ...,\alpha_{\ell}$ be distinct elements of $Q$. Consider the following $t$-FP code $C=\{(f(\alpha_1),f(\alpha_2),...,f(\alpha_{\ell}): f\in Q[X] \text{ and } \deg f < \lceil \ell/t\rceil\}$ (see \cite{SimonF03}). Then, for any integer $T>t$,  there exists a $(T, t)$-FP code $C^{\prime}$ of the same cardinality as $C$ with $q$ groups, where $q=|Q|$.

\begin{proof} Partition $C$ into $q$ groups, $C_1, C_2, ..., C_q$, by letting  $C_i=\{x \in C:  x_1=i\}$. Since $|C|=q^{\lceil \ell/t\rceil}$ and $C^{\prime}=\displaystyle\bigcup_{i=1}^q C_1$, now we have a code $C^{\prime}=C$ containing $q$ groups of cardinality $q^{\lceil \ell/t\rceil-1}$.  It is easy to see that none of the codewords in $C^{\prime}$ can frame the codeword outside its own group, since any pair of codewords from the different groups have different symbols in the first co-ordinate. Together with $t$-FP property that $C^{\prime}$ inherits directly from $C$, $C^{\prime}$ is a $(T, t)$-FP code of the same cardinality as $C$.
\end {proof}
\end{example}
Example $\ref{sepa}$ has shown that there exist $(T, t)$-FP codes as big as $(T, t)$-FP codes in some cases. In fact, for any $t$-IPP, $t$-SFP or $t$-FP code $C$, if there exists an $i \in \{1,2,...,\ell\}$ that only $g \leq q$ symbols from $Q$ are used in the $i$th co-ordinate and they are uniformly distributed, we can construct two-level codes by partitioning its codewords into $g$ groups according to the $i$th co-ordinate, then obtain a two-level code $C^{\prime}$ with the same cardinality as $C$. We exclude full proof from this paper: see \cite{thesis} for details. 

There are one-level codes that the distribution of alphabet symbols in any co-ordinate is non-uniform (and there are examples that are the largest known for some parameters) In this case, the above simple two-level construction cannot be used. The construction we propose in the next subsection is general enough to work in these cases, though at a cost of reducing the size of the code by a factor of up to 2.

\subsection{Main Theorem}
In this section, we aim to construct two-level fingerprinting codes from existing one-level fingerprinting codes. Our construction works provided the number of groups is at most the alphabet size. Our construction begins with a one-level code, removes some codewords, groups and modifies the remaining codewords. The results are two-level codes with are guaranteed to be at least half the size of their original codes. 

The next theorem is the core of our construction.
\begin{theorem}\label{2levelcon}
Let $q, g$ and $\ell$ be integers greater than 1, where $g \leq q$. Let $C$ be an $q$-ary length $\ell$ 
code. Then there exists a $q$-ary length $\ell$ 
code $C^{\prime}$ of cardinality at least $\frac{|C|}{2}$, where $C^{\prime}$ possesses the following properties;
\begin{enumerate}
\item there exists an injection from $C^{\prime}$ to $C$ with changes occur only in the first co-ordinate of the codewords, 
\item $C^{\prime}$ can be partitioned into $g$ groups of the same size, each at least $\left\lceil\frac{|C|}{2g}\right\rceil$,
\item the first co-ordinate of codewords in each group of  $C^{\prime}$ are distinct from those of any other group.
\end{enumerate}
\end{theorem}

The explicit construction of two-level codes is embedded in the proof of Thorem \ref{2levelcon}. Before showing the detailed proof, we provide the following example to give a rough idea about how to construct $C^{\prime}$ from a given code $C$ satisfying Theorem \ref{2levelcon}.

\begin{example}\label{2con}
Let $n=91$, $q=11$, $g=9$ and $p=\left\lceil\frac{n}{2g}\right\rceil=6$. Let $g_1, g_2,..., g_q$ be number of codewords of $C$ beginning with each symbol. Suppose we have:
\begin{align*}
g_1&=4 &
g_2&=5 &
g_3&=10 &
g_4&=11\\
g_5&=17 &
g_6&=5 &
g_7&=2 &
g_8&=4\\
g_9&=18 &
g_{10}&=10 &
g_{11}&=5.
\end{align*}
Our aim is to form new $9$ groups of size $6$, where  the first co-ordinate of codewords in each new group are distinct from any other groups. Our method can be divided in to 3 main steps: splitting, amalgamating and replacing. {We illustrate the construction of $C^{\prime}$ in Figure} \ref{tablefig}.

\noindent\textbf{Step 1: Splitting}\\ 
The purpose of this step is to split all big groups, which contain $p$ codewords or more, into one or more smaller groups of size at least $p$.
Observe that $g_3,g_4$ and $g_{10}$ provide 1 group each, while $g_5$ and $g_9$ give 2 and 3 groups, respectively. 
We now have 8 groups with only 5 different first co-ordinates. Remove any 3 other groups to obtain 3 more first co-ordinates, let them be $g_2$, $g_6$ and $g_{11}$. Hence, 2, 6 and 11 have become unused  first co-ordinates.

\noindent\textbf{Step 2: Amalgamating}\\ In this step, we aim to create more groups, of size at least $p$, by merging at least 2 smaller groups together.
So, we merge the remaining groups together to obtain the last group of size $g_1+ g_7+g_ 8=10$. 

\noindent\textbf{Step 3: Replacing}\\Since some  first co-ordinates of the groups we have constructed are repeated, we replace them with those unused first co-ordinates in this step. Then, reduce the number of codewords in each group to $p$. Here the repeated  first co-ordinates are 5 and 9, replace the first co-ordinate of codewords in those "repeated" groups (1 group from $g_5$ and 2 groups from $g_9$) by the unused  first co-ordinates; 2, 6, 11.

The result after completing these 3 steps and removing any extra elements is $9$ disjoint groups of size $6$, where the first co-ordinate of codewords in each group are distinct from any other groups. The following table illustrates what we did earlier. We abuse the notation and use $k_i$ for the $i$th codeword in group $t$, i.e. the $i$th codeword beginning with symbol $t$.

\begin{figure}
\begin{center}  
\includegraphics[scale=0.80]{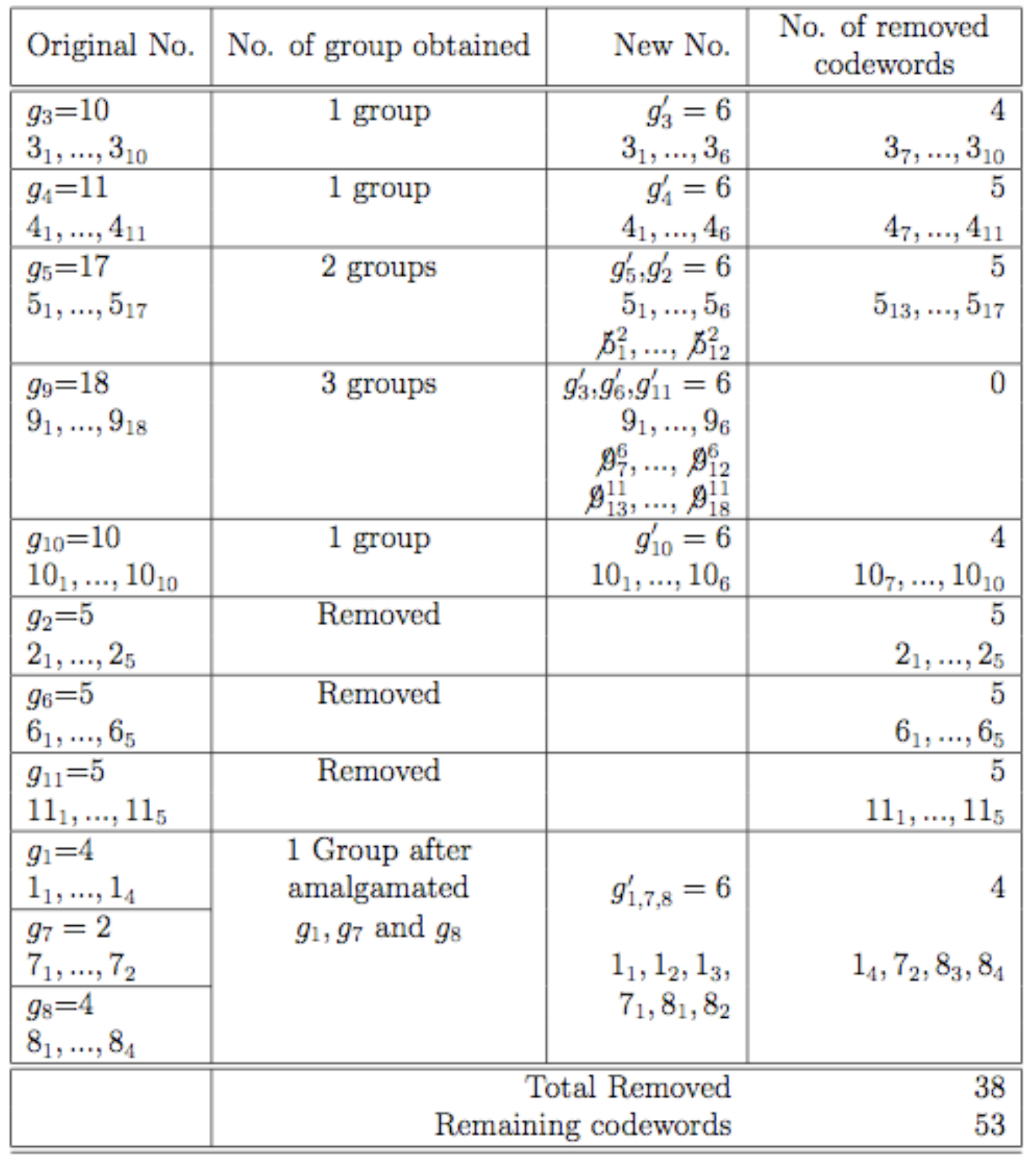}  
\end{center} 
\caption{Groups Dividing in Example \ref{2con}}  
\label{tablefig}
\end{figure}
\end{example}

Now, we prove Theorem \ref{2levelcon} using the similar approach as in Example \ref{2con}.
\begin{proof}[Proof of Theorem 2]
Let $q$ and $\ell$ be integers greater than 1. Let $C$ be a one-level code length $\ell$ over an alphabet $Q$, where $|Q|=q$. Let $g$ be a positive integer less than or equal to $q$. Denote $\left\lceil\frac{|C|}{2g}\right\rceil$ by $p$. For each symbol $a \in Q$, let $G_a=\{x \in C: x_1=a\}$ and denote the size of $G_a$ by $g_a$. And let $g_a=\alpha_a p + \beta_a$, where $\alpha_a,\beta_a$ are integers and $0 \leq \beta_a < p$. Let $Q_1$ be $\{a \in Q: \alpha_a > 0\}$, $q_1= |Q_1|$ and $v=\displaystyle\sum_{a \in Q} \alpha_a=\displaystyle\sum_{a \in Q_1} \alpha_a$.  We are now trying to construct $g$ groups of $p$ codewords which the first co-ordinate of each group is different from the others. As in Example \ref{2con}, we proceed through the three main steps: splitting, amalgamating and replacing.
 
\noindent\textbf{Step 1: Splitting}\\
Now, for each $a$ in $Q_1$, we pick $\alpha_ap$ codewords from each $G_a$, then divide these codewords into $\alpha_a$ sets of $p$ codewords.
At this stage, we obtain $v$ disjoint sets of $p$ codewords, namely $C_1, C_2, ..., C_v$, with the property that all the codewords within the same set have the same symbol in the first co-ordinate.  However, some of the symbols are still being used by more than one group.

If $v \geq g$ we are done: for $i \in \{1,2, ..., g\}$, we replace the first co-ordinate of the codewords in $C_i$ by symbol $i$ to form subset $C^{\prime}_i$, and define $C^{\prime} = \displaystyle\bigcup_{i=1}^g C^{\prime}_i$. So without loss of generality, we assume $v<g$. To construct the first $v$ groups from $C_1, C_2, ..., C_v$
, we need $v$ different symbols to replace the first co-ordinate of each group. Besides the $q_1$ symbols in $Q_1$, we need $v-q_1$ extra symbols.
Let $Q_2$ be any subset of $Q$ of cardinality $v$ containing $Q_1$. Throw away all codewords in $G_a$ where $a \in Q_1\backslash Q_2$.

\noindent\textbf{Step 2: Amalgamating}\\
We merge some of the remaining $G_a$, where $a \in Q \backslash Q_1$, into $g-v$ groups of size between $p$ and $2p-2$, where the first co-ordinates of each group are different from the other groups. This can be done as follows. 


Consider
\begin{align*}
\displaystyle\sum_{a \in Q\backslash Q_2} g_a &= |C|-\displaystyle\sum_{a \in Q_2} g_a\\
&= |C|-\displaystyle\sum_{a \in Q_2} \left(\alpha_ap + \beta_a\right)\\
&=|C|- \left(\displaystyle\sum_{a \in Q_1} \alpha_ap + \displaystyle\sum_{a \in Q_2} \beta_a\right)\\
&\geq|C|- \left(vp + v(p-1)\right)\\
&\geq 2gp-vp-v(p-1)\\
&=2(g-v)p+v.
\end{align*}
Hence, apart from 
$\displaystyle\cup_{a \in Q_2}G_a$, we have at least $2(g-v)p+v > 2(g-v)p$ codewords left in the system. And, since $g_a=\alpha_a + \beta_a=0 + \beta_a=\beta_a \leq p-1$ for all $a \in   Q\backslash Q_2$, we can group $g_a, a \in Q\backslash Q_2$ in a greedy fashion into $g-v$ sets of size between $p$ and $2p-2$, where each $a \in Q\backslash Q_2$ is allowed to merge into at most one set. Name the sets $C_{v+1}, C_{S+2}..., C_g$. Observe that the first co-ordinate of codewords in each set varies, but differs from any other groups.
 
\noindent\textbf{Step 3: Replacing}\\
Here we construct  the groups of 
$C^{\prime}$ 
as follows: Let $Q_2=\{a_1, a_2, ..., a_g\}$. 
\begin{enumerate}
	\item for $i=1$ to $v$, let $C_i^{\prime}$ be a set of codewords obtained from $C_i$ by replacing the first co-ordinate by the symbol $a_i \in Q_2$,
	\item for $i=v+1$ to $g$, let $C_i^{\prime}$ be a set of any $p$ codewords from $C_i$,
         \item let $C^{\prime} = \displaystyle\bigcup_{i=1}^g C^{\prime}_i$.
\end{enumerate}

Now, we need to show that our constructed code $C^{\prime}$ satisfies Theorem  \ref{2levelcon}.

Let the mapping $\varphi:C^{\prime} \longrightarrow C$ map 
each codeword of $C^{\prime}$ to its corresponding codeword in $C$. It is not difficult to see that $\varphi$ is an injection that makes changes in only the first co-ordinate of any codeword.

Hence, we are now obtain a code $C^{\prime}$ contains $g$ groups of the same size, each at least $\left\lceil\frac{|C|}{2g}\right\rceil$ with the property that the first co-ordinate of codewords in each group are distinct from any other groups, and an injection $\varphi$ from $C^{\prime}$ to $C$ with changes occur only in the first co-ordinate of the codewords.

Note that to construct $C^{\prime}$, we have eliminated $\left(\displaystyle\sum_{a \in Q} \beta_a\right) - (g-v)p$ codewords from $C$. Here the first term represents all the remainders, and the second term is derived form $(g-v)$ amalgamated groups that were taken back from the thrown away remainders. Now  $\left(\displaystyle\sum_{a \in Q} \beta_a\right) - (g-v)p = (|C|-vp) - (g-v)p = |C| - gp \leq |C| -\frac{|C|}{2}=\frac{|C|}{2}$. Hence we eliminate at most $\frac{|C|}{2}$ codewords from $|C|$.

\end{proof}

For any one-level FP, SFP and IPP code, the two-level code satisfying Theorem \ref{2levelcon} has the corresponding two-level fingerprinting property. To make it more convenient for us to show this in the next subsection, we define some mappings and a lemma we need here. 

Let the mapping $\pi:Q \longrightarrow Q$ be defined as follows.  Let $\pi(a)=a$ when $a$ does not appear as the first co-ordinate of any codeword of $C^{\prime}$: otherwise 
let $\pi(a)=b \in Q$ when there exists a codeword $c^{\prime} \in C^{\prime}$ with $c^{\prime}_1=a$ that was derived from $c \in C$ with $c_1=b$. 

Let the mapping $\psi:Q^{\ell} \longrightarrow Q^{\ell}$ be defined by mapping $x\in Q^{\ell}$ to $\psi(x)\in Q^{\ell}$ where
\begin{align*}
\psi(x)_i = \begin{cases} {\pi(x_i)} &\text{if }i=1; \\ x_i &\text{otherwise.} \end{cases} 
\end{align*} 
It is not difficult to see that $\psi$ is a well-defined function and $\varphi$ from Theorem \nolinebreak \ref{2levelcon} is actually $\psi$ when restricted to $C^{\prime}$, i.e. $\varphi=\psi|_{C^{\prime}}$.

Observe that for any $i \in \{1, 2, ..., \ell\}$ and any codewords $y,z \in C^{\prime}$, if $y_i=z_i$, then $\varphi(y)_i=\varphi(z)_i$. Moreover, $\varphi(y)_i=y_i=z_i=\varphi(z)_i$ when $i \not = 1$.

\begin{lemma}\label{srb}Let $C$ and $C^{\prime}$ be codes length $\ell$ over $Q$ satisfying Theorem \ref{2levelcon}. Let $X$ be a subset of $C^{\prime}$. Then
\begin{align*}
\psi(\desc(X))\subseteq\desc(\psi(X)).
\end{align*}
\end{lemma}
\begin{proof}
Let $X$ be a subset of $C^{\prime}$ and let $y$ be a codeword in $\psi(\desc(X))$. Then, there exists a codeword $x$ in $\desc(X)$ such that $\psi(x)=y$. For any component $i$ in $\{1, 2, ..., \ell\}$, there exists a codeword $x^i$ in $X$, where $x_i=x^i_i$ . Hence $\psi(x)_i=\psi(x^i)_i$ for all $i \in \{1, 2, ..., \ell\}$. Which implies $\psi(x) \in \desc(\{\psi(x^1), \psi(x^2), ..., \psi(x^{\ell})\}) \subseteq \desc(\psi(X))$. Therefore $y \in \desc(\psi(X))$, which implies $\psi(\desc(X))\subseteq\desc(\psi(X))$.
\end{proof}
%
%
\subsection{Existence of Codes}\label{seccon}
Here we demonstrate that the codes $C^{\prime}$ satisfying Theorem \ref{2levelcon} are two-level FP, SFP or IPP code if the original codes $C$ are respecting FP, SFP or IPP codes. Also, we provide an example showing that the two-level code constructed from  a TA code using Theorem \ref{2levelcon} does not always possess two-level TA property.
\begin{theorem}\label{2lfp}
Let $t, q, g$ and $\ell$ be integers greater than 1, where $g \leq q$. And let $T$ be any integer greater than $t$. Suppose there exists a $q$-ary length $\ell$ one-level $t$-FP code $C$. Then there exists a $q$-ary length $\ell$ two-level $(T, t)$-FP code $C^{\prime}$ of cardinality at least $\frac{|C|}{2}$, where $C^{\prime}$ contains $g$ groups of the same size.
\end{theorem}
\begin{proof}
Let $C^{\prime}$ be a code obtained from the $t$-FP code $C$ as in Theorem \nolinebreak \ref{2levelcon}.
It is easy to see that none of the codewords in $C^{\prime}$ can frame the codeword outside its own group, since any pair of codewords from the different groups have different symbols in the first co-ordinate. So, only $t$-FP property of $C^{\prime}$ needs to be proved.

Let $U$ be any subset of $C^{\prime}$ containing at most $t$ codewords. Let $x \in \desc(U)\cap C^{\prime}$. We will show that $x\in U$. Since $x \in \desc(U)\cap C^{\prime}$, we have $\varphi (U)\subseteq C$ and $|\varphi (U)| \leq |U| \leq t$.

Since $x \in \desc(U)\cap C^{\prime}$, then  $x \in \desc(U)$ and $x \in C^{\prime}$. By Lemma\ref{srb}, $\varphi(x) \in \desc(\varphi(X))$. Also it is easy to see that $\varphi(C^{\prime}) \subseteq C$, hence $\varphi(x) \in \desc(\varphi(X)) \cap C$. This makes $\varphi(x) \in \varphi(U)$ by the $t$-FP property of $C$. Hence $x\in U$, which implies $C^{\prime}$ has $t$-FP property, i.e. $C^{\prime}$ is a $(T,t)$-FP code.

%
Then we can conclude that there exists a $q$-ary length $\ell$ two-level $(T,t)$-FP code $C^{\prime}$ of size at least $\frac{|C|}{2}$, containing $g$ groups (each of size at least $\left\lceil\frac{|C|}{2g}\right\rceil$).
\end{proof}
 
\begin{theorem}\label{2lsfp}
Let $t, q, g$ and $\ell$ be integers greater than 1, where $g \leq q$. And let $T$ be any integer greater than $t$. Suppose there exists a $q$-ary length $\ell$ one-level $t$-SFP code $C$. Then there exists a $q$-ary length $\ell$ two-level $(T, t)$-SFP code $C^{\prime}$ of cardinality at least $\frac{|C|}{2}$, where $C^{\prime}$ contains $g$ groups of the same size.
\end{theorem}
\begin{proof}

Let $C^{\prime}$ be a code obtained from the $t$-SFP code $C$ as in Theorem \nolinebreak \ref{2levelcon}.
 
\begin{enumerate}
\item Let $X_1, X_2$ be subsets of $C$ of size at most $t$, where $\desc(X_1) \cap \desc(X_2) \not = \emptyset$. We will show that $X_1 \cap X_2 \not = \emptyset$.
 
 Let $x \in \desc(X_0) \cap \desc(X_1)$. Then $\varphi(x) = \psi(x) \in \psi(\desc(X_0) \cap \desc(X_1)).$ Now 
\begin{align*}
\varphi(x) &\in \psi(\desc(X_1) \cap \desc(X_2))\\
&\subseteq \psi(\desc(X_1)) \cap \psi(\desc(X_2))\\
&\subseteq \desc(\psi(X_1))\cap \desc(\psi(X_2)) \text{ by Lemma } \ref{srb}\\
&= \desc(\varphi(X_1))\cap \desc(\varphi(X_2)) \text{}. 
\end{align*}
Therefore $\desc(\varphi(X_1))\cap \desc(\varphi(X_2)) \not = \emptyset$. By the $t$-SFP property of $C$, we deduce that $\varphi(X_1) \cap \varphi(X_2) \not = \emptyset$. Since $\varphi$ is an injection, $\varphi(X_1) \cap \varphi(X_2)=\varphi(X_1 \cap X_2)$. Therefore $X_1 \cap X_2 \not = \emptyset$, which implies $C^{\prime}$ has the $t$-SFP property.
 
\item Let $Y_1, Y_2$ be subsets of $C$ of size at most $T$, where $\desc(Y_1) \cap \desc(Y_2) \not = \emptyset$. We will show that $\mathcal{G}(Y_0) \cap \mathcal{G}(Y_1) \not = \emptyset$.
 
 Let $x \in \desc(Y_1) \cap \desc(Y_2)$. Then there exist codewords $a$ in $Y_1$ and $b$ in $Y_2$, where $a_1=x_1=b_1$. Since the first co-ordinate of each group is different from the others, we can conclude that $\mathcal{G}(a) = \mathcal{G}(b)$. Therefore $\mathcal{G}(a) \in \mathcal{G}(Y_1) \cap \mathcal{G}(Y_2) \not = \emptyset$, so $C^{\prime}$ is a $(T,t)$-SFP code.
\end{enumerate}
\end{proof}

\begin{theorem}\label{2lipp}
Let $t, q, g$ and $\ell$ be integers greater than 1, where $g \leq q$. And let $T$ be any integer greater than $t$. Suppose there exists a $q$-ary length $\ell$ one-level $t$-IPP code $C$. Then there exists a $q$-ary length $\ell$ two-level $(T, t)$-IPP code $C^{\prime}$ of cardinality at least $\frac{|C|}{2}$, where $C^{\prime}$ contains $g$ groups of the same size.
\end{theorem}
\begin{proof}
Let $C^{\prime}$ be a code obtained from the $t$-IPP code $C$ as in Theorem \nolinebreak \ref{2levelcon}.

\begin{enumerate}
\item Let $x \in \desc_{t}(C^{\prime})$. Then, there exists $U \subseteq C^{\prime}$ such that $|U| \leq t$ and $x \in \desc(U)$. By Lemma \ref{srb}, we know that $\psi(x) \in \desc(\varphi(U))$ and $\varphi(U) \subseteq C$. Observe that $|\varphi(U)|\leq|U|\leq t$. Hence $\psi(x) \in \desc_{t}(C)$. Since $C$ is an IPP code, 
\begin{align*}
\displaystyle\bigcap_{\substack{X \subseteq C : |X| \leq t \\ \psi(x) \in \desc(X)}} X \not = \emptyset.
\end{align*} 
Also, for any $X \subseteq C^{\prime}$, $x \in \desc(X)$ implies $\psi(x) \in \desc(\varphi(X))$ and $|X|=|\varphi(X)|$. Hence
\begin{align*}
\displaystyle\bigcap_{\substack{X \subseteq C : |X| \leq t \\ \psi(x) \in \desc(X)}} X \subseteq \displaystyle\bigcap_{\substack{X \subseteq C^{\prime} : |X| \leq t \\ x \in \desc(X)}} \varphi(X).
\end{align*}
Since $\varphi$ is injective, we have
\begin{align*}
\displaystyle\bigcap_{\substack{X \subseteq C^{\prime} : |X| \leq t \\ x \in \desc(X)}} \varphi(X) = \varphi(\displaystyle{\bigcap_{\substack{X \subseteq C^{\prime} : |X| \leq t \\ x \in \desc(X)}} X}).
\end{align*}
Hence
\begin{align*}
\varphi(\displaystyle{\bigcap_{\substack{X \subseteq C^{\prime} : |X| \leq t \\ x \in \desc(X)}} X})\not = \emptyset.
\end{align*}
Therefore
\begin{align*}
\displaystyle{\bigcap_{\substack{X \subseteq C^{\prime} : |X| \leq t \\ x \in \desc(X)}} X}\not = \emptyset,
\end{align*}
which shows that $C^{\prime}$ is a $t$-IPP code.
 
\item Let $y \in \desc_{T}(C^{\prime})$. Then, there exists $V \subseteq C^{\prime}$ such that $|V| \leq T$ and $y \in \desc(V)$. Let $v$ be a codeword in $V$ such that $v_1=y_1$, and let $i \in \{1, 2, ..., g\}$ such that $v \in C^{\prime}_i$. Hence $\mathcal{G}(v)=i$. For any $X \subseteq C^{\prime}$ of cardinality at most $T$ with $\desc(X)$ containing $y$, there exists a codeword $y^X$ such that $y_1=y^X_1$. Since the group index of a codeword can be determined from its first co-ordinate, we have $\mathcal{G}(y^X)=i$. That implies
\begin{align*}
i \in \displaystyle\bigcap_{\substack{X \subseteq C^{\prime} : |X| \leq T \\ y \in \desc(X)}} \mathcal{G}(X).
\end{align*}
Consequently,
\begin{align*}
\displaystyle\bigcap_{\substack{X \subseteq C^{\prime} : |X| \leq T \\ y \in \desc(X)}} \mathcal{G}(X) \not = \emptyset.
\end{align*}
\end{enumerate}
Therefore, $C^{\prime}$ is a $(T,t)$-IPP code.
\end{proof}

The two-level codes satisfying Theorem \ref{2levelcon} preserve the fingerprinting property from their corresponding one-level codes for IPP, SFP and FP codes. However, this is not always true in the case of TA codes as can be seen in the following example.

\begin{example}
Let $C =\{011, 022, 033, 044, 105, 206, 307, 408, 550, 660, 770, 880\} \subseteq \{0, 1,$ $2, ..., 8\}^3$. It is not difficult to check that $C$ is a $2$-TA code. Let $g=4$, then $p=\left\lceil\frac{12}{8}\right\rceil=2$. Then Theorem \ref{2levelcon} does not guarantee two-level traceability code from $C$.
\end{example}
\begin{proof}
Here we have $g_1=4$, $g_2=g_3=...=g_8=1$. Consider $C_1=\{011, 022\}$, $C_2=\{033, 044\}$, $C_3=\{105, 550\}$ and $C_4=\{206, 660\}$, which leads to $C^{\prime}_1=\{011, 022\}$, $C^{\prime}_2=\{833, 844\}$, $C^{\prime}_3=\{105, 550\}$ and $C^{\prime}_4=\{206, 660\}$ by Theorem \ref{2levelcon}. Let $U=\{011,105, 550\}$, then $000 \in \desc(U)$ and $\mathcal{G}(U)=\{1,3\}$. Observe that $206$ is a codeword of $C^{\prime}$ with $d_H(000,206)$ minimal, but $\mathcal{G}(206)=4 \not \in \mathcal{G}(U)$. Therefore $C^{\prime}$ is not a $(K,2)$-TA code for any integer $T$ greater than 2.
\end{proof}
%
\section{Conclusion and Open problems}
Theorem 2 ensures that we can always construct two-level IPP, SFP and FP codes, with $g\leq q$, of size at least half of the size of one-level codes. When one-level code of exponential size, throwing away half of the codewords would not effect the codes' size significantly. However, we do not have the same result for TA codes. Hence the following question comes up naturally. 
\\\\
\textbf{Question 1:} Let $g \leq q$, and let $C$ be a $q$-ary $t$-TA code of length $\ell$, does there always exist a $q$-ary $(T, t)$-TA code $C^{\prime}$ of length $\ell$ of cardinality at least a half of the original code $C$, containing $g$ groups? 
\\\\
We believe the answer to this question is yes. 

The results in this paper require the number of groups to be small. 
\\\\
\textbf{Question 2:} Are there any good constructions of two-level fingerprinting codes when the number of groups is greater than the alphabet size? 
\\\\
And the next question follows.
\\\\
\textbf{Question 3:} Are there any good upper bounds on the size of two-level codes that are significantly better than the one-level case? 
\\\\
We believe that, in general cases, the bounds will be significantly better than the one-level case and will begin to depend on $T$ when $g$ grows bigger.

\bibliographystyle{plain}	
\bibliography{120109paper}	
\end{document}